\documentclass[a4paper,fleqn]{cas-sc}
\usepackage{multicol}
\usepackage{tikz}
\usepackage{empheq}
\newcounter{casenum}

\usetikzlibrary{patterns,positioning}
\usetikzlibrary{shapes.geometric,calc,decorations.text}
\usepackage{caption}
\usepackage{subcaption}
\usepackage{todonotes}
\usepackage{bbold} 
\usepackage{float}
\usepackage{adjustbox}
\usepackage[authoryear,longnamesfirst]{natbib}
\usepackage{thmtools}
\usepackage{amsmath}
\usepackage{amsfonts,amssymb,amsmath}
\usepackage{xcolor}
\usepackage{url}
\usepackage{algorithm}
\usepackage[justification=centering]{caption}
\usetikzlibrary{patterns,positioning}
\usepackage{caption}

\usepackage{amsmath}
\usepackage{amsfonts,amssymb,amsmath,amsthm}

\usetikzlibrary{shapes.misc}

\usetikzlibrary{shapes.geometric}
\usetikzlibrary{patterns,positioning}
\usepackage{amsmath, amssymb}
\usepackage{tikz}
\usetikzlibrary{shapes.geometric,calc,decorations.text}
\usetikzlibrary{patterns}

\usepackage{mathrsfs}
\usepackage{amssymb}
\usepackage{dsfont}
\usetikzlibrary{patterns}
\usetikzlibrary{shapes.geometric}
\def\mytransformation{
	\pgfmathsetmacro{\myX}{0.9\pgf@x}
	\pgfmathsetmacro{\myY}{0.0045*(\pgf@x)*(\pgf@x)+(\pgf@y)}
	\setlength{\pgf@x}{\myX pt}
	\setlength{\pgf@y}{\myY pt}
}

\makeatother
\usetikzlibrary{backgrounds}
\usetikzlibrary{shapes.geometric}
\usetikzlibrary{calc}
\usepackage{enumerate}
\usepackage{mathrsfs}
\usepackage{amssymb}
\usepackage{graphicx}%
\usepackage{multirow}%
\usepackage{amsmath,amssymb,amsfonts}%
\usepackage{amsthm}%
\usepackage{mathrsfs}%
\usepackage{xcolor}%
\usepackage{textcomp}%
\usepackage{manyfoot}%
\usepackage{algpseudocode}%
\usepackage{setspace}
\usepackage{pdflscape}
\usepackage{graphicx}
\usepackage{pst-node}
\usepackage{booktabs}
\usepackage{tikz}
\usepackage{todonotes}
\usepackage[inline]{enumitem}

\renewenvironment{proof}{{\it Proof.}}{\qed \medskip}
\newtheorem{theorem}{Theorem}[section]

\newtheorem{obs}[theorem]{Observation}

\newtheorem{claim}[theorem]{Claim}
\newtheorem{conjecture}[theorem]{Conjecture}
\newtheorem{lemma}[theorem]{Lemma}

\newtheorem{corollary}[theorem]{Corollary}

\usepackage{mathtools}

\begin{document}
	
	\title{Extended formulations for induced tree and path polytopes of chordal graphs }
	\shorttitle{induced tree and path polytopes of chordal graphs}
	
	\author[1,2]{Alexandre Dupont-Bouillard}[orcid = 0009-0007-4207-8948]
	\ead{dupont-bouillard@lipn.fr}
\shortauthors{Alexandre Dupont-Bouillard}
	\affiliation[1]{organization = {IRISA CNRS UMR 6074, Université de Rennes}, 
		adressline={263 Av. Général Leclerc}, 
		city={Rennes},
		 postcode={35000},
		 country={France}}
	
	\affiliation[2]{
		organization={LIX -- CNRS UMR 7161, Institut polytechnique de Paris}, 
		adressline = {route de Saclay}, 
		city= {Palaiseau}, 
		postcode= {91120},
		country = {France}}

	\begin{abstract}		
		In this article, we give an extended space formulation for the induced tree polytope and another one for the induced path polytope of chordal graphs, both with vertex and edge variables.
		These formulations are obtained by proving that, in chordal graphs, the extreme points of those polytopes respectively form two Hilbert bases.
				
		While the formulation for the induced tree polytope has a polynomial size, the formulation we propose for the induced path polytope involves an exponential number of inequalities. We identify which of these inequalities are facet‑defining and exhibit a polynomial time separation algorithm based on computing a minimum $s$-$t$ cut in an auxiliary graph.
		
		As corollaries, we obtain that the problems of finding an induced tree or path maximizing a linear function over the edges and vertices are solvable in polynomial time for the class of chordal graphs.  
	\end{abstract}
	

	\maketitle

	\section*{Motivations }
	
	In this article, we study the problems of finding maximum weight induced trees and paths where both the vertices and the edges are weighted. In particular, we fill a gap in the polyhedral literature associated with these problems. 
	
	\medskip
	\paragraph{Induced trees}
	The problem of finding the size of the largest induced tree in a graph is a problem that has been approached through extremal combinatorics~\cite{ERDOS198661,MATOUSEK2007307,FOX2009494}. In particular, an important question is that of bounding this number according to several graph parameters, such as the number of vertices, edges, radius, independence number, maximum clique, and connectivity.
	In~\cite{ERDOS198661}, a Ramsey-type bound is proposed on the minimum integer $n(k,t)$ such that every graph having at least $n(k,t)$ vertices contains either a clique of size $k$ or an induced tree of size $t$. Finally, they show that every connected graph with radius $r$ has an induced tree of size at least $2r-1$ by exhibiting a large induced path.  
	A related problem is the one of finding large induced forests~\cite{inducedforest} in which bounds for the largest induced forest are given with respect to maximum degree or stability number. Besides these studies, no polyhedral study of the induced tree polytope on particular classes has been proposed.
	The problem of computing the maximum size of an induced tree is known to be NP-hard~\cite{completePathTree} in general, but it is also the case in bipartite and triangle-free graphs with maximal degree 3~\cite{MATOUSEK2007307}. 
	Several integer linear programs have been proposed for the maximum weight induced tree problem~\cite{ILPinducedtree} based on different variable spaces, and extensive computational experiments have been made on general graphs.
	\medskip
	
	\paragraph{Induced paths}
	
	The problem of finding a longest induced path has applications in assessing the robustness of a network. It can be viewed as a variant of the graph diameter. In a graph where vertices represent communication devices that may fail, and edges represent possible transmission links, the length of the longest induced path indicates the worst-case number of intermediate devices required for a one-to-one communication between any two vertices $u$ and $v$ after the failure of a set of devices that do not disconnect $u$ from $v$.
	Another application is about error correcting codes~\cite{error}. 
	The induced paths of a graph are induced trees, and the type of questions raised in the previous paragraph can be raised considering this structure. The authors of~\cite{FOX2009494} show that every connected triangle-free graph on $n$ vertices has an induced path of length at least $\sqrt{n}$, and their approach can be generalized to obtain lower bounds for graphs with no clique of size $k$.
	In~\cite{esperet}, the authors show that every 3-connected planar graph contains an induced path on $\text{log}(|V|)$ vertices. In~\cite{jfraymond}, it is observed that a graph may contain a large path but no large induced path; for this reason, they characterize the sets of edges to be forbidden along a path of length $n$ to force the existence of an induced path of length $log(n)$. This result allows them to show that if a graph has a path of length $n$ and no $K_t$ as a topological minor, then it contains an induced path of order $log(n)^{\Omega (1/t log^2 t)}$.
	
	Finding a maximum vertex weighted induced path is also NP-hard~\cite{completePathTree}, but 
	polynomially solvable in $k$-chordal graphs, interval filament,  graphs decomposable by clique cutset or by splits whose prime subgraphs are polynomially solvable~\cite{GAVRIL2002203}, improved in~\cite{HANSEN2009135} and on graphs of bounded clique width, asteroidal triple free graphs~\cite{inducedpathcliquewifth}.
	Besides these combinatorial algorithms, no polyhedral study of the corresponding polytope has been made. A related polyhedra has however, been considered in~\cite{goel2025halfspacerepresentationspathpolytopes}: the polytope whose extreme points in $\{0,1\}^{E}$ encode paths between pairs of leaves in a tree. In trees, induced path and path are the same subgraphs, so the latter article describes the induced path (between leaves) polytope of trees with edge variables. Finally, extended space integer linear programs have been proposed for the maximum weight induced path~\cite{ILPinducedPath,ILPinducedPath2}, leading to computational experiments. 
	
	\section*{Contributions}
	
	In this article, we study both the induced tree polytope and the induced path polytopes of chordal graphs in two separate sections. 
	It turns out that beside being included one in the other, these two polytopes share similar properties. 
	Both following sections follow the same scheme, the first is about induced trees, and the second is about induced path.
	Each section starts by showing that the extended incidence vectors of their corresponding graph structure form a Hilbert basis in chordal graphs. As a byproduct, we obtain linear systems of inequalities describing the cones they generate in chordal graphs. 
	Finally, we show that intersecting these cones with the correct hyperplane yields the induced tree and path polytope of chordal graphs in extended variable space.
		 While the system we provide for the induced tree polytope is easily seen to be compact in chordal graphs, the one we provide for the induced path polytope has an exponential number of inequalities. 
	 Therefore, we characterize the ones that define facets and provide a polynomial separation algorithm for those.
	 
	 This study allows to solve the problems of finding a maximum vertex and edge weight induced tree or a maximum weight induced path in polynomial time for chordal graphs. Up to now, only the maximum vertex weighted induced path problem was known to be polynomially solvable on chordal graphs.

		\section*{Definitions}
	All the graphs in this article are simple. 
	Given a graph $G=(V,E)$, we denote its \textit{complement} by $\overline{G}=(V,\overline E)$, where $\overline{E} = \{ e \in \binom{V}{2} : e \notin E \}$. 
	We denote by $V(G)$ and $E(G)$ the vertex set and the edge set of $G$, respectively. Two vertices $u$ and $v$ are \emph{adjacent} if $uv \in E(G)$.
	Given a subset of vertices $W \subseteq V$, we denote by $E(W)$ the set of edges of $G$ having both endpoints in $W$, and by $\delta (W)$ the set of all edges having exactly one endpoint in $W$.
	When $W$ is a singleton $\{w\}$, we simply write $\delta (w)$. We say that the edges in $\delta (w)$ are \textit{incident} to $w$. Given two subsets of vertices $W_1$, $W_2$, we denote by $\delta(W_1,W_2)$ the set of edges having one endpoint in $W_1$ and the other in $W_2$.
	For $F \subseteq E$, we denote by $V(F)$ the set of vertices incident to any edge of $F$.
	Given $W\subseteq V$, the graph $G[W]=(W,E(W))$ is the \textit{subgraph of $G$ induced by $W$}. When $H$ is an induced subgraph of $G$, we say that $G$ {\em contains} $H$. We denote by $G \setminus v$ (resp. $G\setminus uv$¨) the graph obtained from $G$ by removing the vertex $v$ (resp. edge $uv$).
	Given a vertex $u\in V(G)$, we denote  by $N_G(u) = \{w \in V(G) : uw \in E(G) \}$ its \textit{neighborhood}, and by $N_G[u] = N_G(u) \cup \{ u\}$ its \textit{closed neighborhood}, when clear from context we simply write $N(u)$ and $N[u]$. The neighborhood of a vertex subset $W$, denoted by $N(W)$ is the set of vertices adjacent to at least one element of $W$.
	A vertex $w$ is \emph{complete} to a subset of vertices $W$ if it is adjacent to each vertex of $W$ and does not belong to $W$. 
	A vertex $u$ complete to $V \setminus \{u\}$ is said to be \emph{universal}. Given a set of vertices $W \subseteq V$, we denote by $C_G(W)$ the vertices complete to $W$ in $G$; we simply write $C(W)$ when the graph is clear from context.
	A \emph{clique} is a set of pairwise adjacent vertices. A \emph{stable set} is a set of pairwise nonadjacent vertices. A graph whose vertex set can be partitioned into two stable sets is called \emph{bipartite}.
	A vertex is \emph{simplicial} if its neighborhood induces a clique.
	
	A \textit{path} (resp. \textit{hole}) is a graph induced by a sequence of vertices $(v_1,\dots,v_p)$ whose edge set is $\{ v_iv_{i+1}:  \ i = 1,\dots, p-1 \} $ (resp. $\{ v_iv_{i+1}:  \ i = 1,\dots, p-1 \} \cup \{v_1v_p \}$ and $p \ge 4$). A \emph{cycle} is a set of edges inducing a connected graph with vertices of degree equal to 2. A \emph{tree} is a connected acyclic graph.
	
	A subset of vertices induces a path (resp. hole) if its elements can be ordered into a sequence inducing a path (resp. hole).
	By abuse of definition, we consider that a set of vertices inducing a path is an \emph{induced path}, similarly, an \emph{induced tree} is a set of vertices inducing a tree. 
	The {\em length} of a path or hole is its number of edges. The { \em parity} of a path or hole is the parity of its length. Given a path induced by the sequence $(w_1,\dots,w_k)$, 
	we call $w_1$, $w_k$ its \emph{extremities} and the set of vertices $\{w_2,\dots,w_{k-1}\}$ its interior moreover, we say that $P$ ends in $w_1$ and $w_k$. 
	The \textit{contraction} of an edge $uv$ in $G$ yields a new graph $G/uv$ built from $G$ by deleting $u$, $v$, adding a new vertex $w$, and adding the edges $wz$ for all $z\in N_G(u) \cup N_G(v)$; parallel edges and loops are removed. 
	For $F \subseteq E$, we denote by $G/F$ the graph obtained from $G$ by contracting all the edges in $F$.
	
	A graph is $k$-chordal if its holes have maximum size $k$; graphs with no holes are called \emph{chordal}. Equivalently, a graph is chordal if there exists an order $(v_1,\dots,v_k)$ of its vertices such that $v_i$ is simplicial in $G_i = G[\{v_i,\dots,v_k \}]$; such an order is called a \emph{perfect elimination order}. Moreover, the maximal cliques of $G$ are among $ \{N_{G_i}[v_i] : i \in \{1,\dots,k\} \}$, which shows that their number is linear.

	\medskip
	Let $\mathcal{P} = \{x:Ax \le b\}\subseteq \mathds{R}^n$ be a \emph{polyhedron}, that is, the intersection of finitely many half-spaces. The \emph{dimension} of $\mathcal{P}\subseteq \mathds{R}^n$, denoted by $\textrm{dim}~\mathcal{P}$, is the maximum number of affinely independent points in $\mathcal{P}$ minus one. 
	If $a \in \mathds{R}^n\setminus \{0\}$, $\alpha \in \mathds{R}$, then the inequality $a^\top x \le \alpha$ is said to be \emph{valid} for $\mathcal{P}$ if $\mathcal{P} \subseteq \{ x \in \mathds{R}^n : a^\top x \le \alpha \}.$ 
	We say that a valid inequality $a^\top x \le \alpha$ \emph{defines a facet} of $\mathcal{P}$ if $\textrm{dim}(\mathcal{P} \cap \{x \in \mathds{R}^n : a^\top x = \alpha \})$ is equal to $\textrm{dim}~\mathcal{P}$ minus one. 
	The polyhedron $\mathcal{P} \subseteq \mathds{R}^n$ is \emph{full dimensional} when dim~$\mathcal{P}$ is equal to $n$.
	In that case, a linear system $Ax \le b$ that defines $\mathcal{P}$ is \emph{minimal} if each inequality of the system defines a facet of $\mathcal{P}$. Moreover, these \emph{facet-defining inequalities} are unique up to scalar multiplication.
	The \emph{size} of a linear system of inequalities is its number of rows.
	A \emph{face} of $\mathcal{P}$ is a polyhedron contained in $\mathcal{P}$ obtained by setting to equality some inequalities in $Ax\leq b$.
	Equivalently, a face is the intersection of the supporting hyperplane of a valid inequality and $\mathcal{P}$.
	\emph{Extreme points} of a polyhedron are its faces of dimension zero.
	Each extreme point satisfies $\textrm{dim}~\mathcal{P}$ linearly independent valid inequalities for $\mathcal{P}$ with equality. 
	A polyhedron $\mathcal{P} \subseteq \mathds{R}^n$ is said to be \emph{integer} if all its non-empty faces contain an integer point.   
	A \emph{polytope} is a bounded polyhedron, equivalently, it can be defined as the convex hull of a finite number of points $P$ and is denoted by $\textrm{conv} (P)$, it is binary if all its extreme points are binary. 
	A \emph{cone} is a polyhedron described by a linear system with a right-hand side equal to 0; it is pointed if it contains no affine space. 
	Given a set of points $P \subseteq \mathds{R}^n$, we denote by $\text{cone} (P)$ the conic hull of $P$'s elements, that is, the minimal cone containing $P$. A set of points $P$ is a generating set of the cone $\mathcal{C}$ if $\text{cone} (P) = \mathcal{C}$.
	The minimal generating sets of pointed cones are unique up to positive scalar multiplication; we call the elements of such a set its \emph{generators} they are in one-to-one correspondence with its faces of dimension 1. 
	 A set of points $X = \{\chi_1,\dots,\chi_n\}$ forms a \emph{Hilbert basis} if every integer point in the conic hull of $X$ is expressible as an integer nonnegative combination of elements of $X$. We say that a finite set of vectors $\mathcal{V}$ \emph{form a Hilbert basis} of a cone $\mathcal{C}$ if the integer points of $\mathcal{C}$ are obtained as natural integer combinations of elements of $\mathcal{V}$ and if $\mathcal{V} \subseteq \mathcal{C}$. 
	\medskip
	
	In this work, the polyhedra we consider are convex hulls of points encoding structured induced subgraphs of an input graph $G=(V,E)$. 
	An induced subgraph $G'=(V',E')$ with $E' = E(V')$, can be encoded by its \emph{incidence vector} $\chi^{V'} \in \{0,1\}^{V}$ whose components associated with elements of $V'$ are equal to 1 and the remaining components equal 0. As induced trees and induced path are naturally described by vertex set, one might call the convex hull of their incidence vectors the natural variable space polytope.
	In this article, we consider the \emph{extended incidence vector} denoted by $\xi^{V'} = (\chi^{V'},\zeta^{E'}) \in \{0,1\}^{V} \times \{0,1\}^E$ where the components of $\zeta^{E'}$ associated with elements of $E $ are equal to 1 and the remaining components equal 0. Considering extended descriptions such as proposed here is a common way to obtain compact polyhedral descriptions of combinatorial polytopes~\cite{exntededfo}. In particular, considering extended spaces sometimes allows to obtain compact descriptions, whereas the natural variable may require exponentially many constraints.
	Given a point $(x,y) \in \mathds{R}^{V} \times \mathds{R}^E$ we denote by  $(x,y)|_{G'}$ its restriction to elements of $G'$.

	Given a graph $G=(V,E)$, the polytopes we consider in this article are defined as follows:
	\begin{itemize}
		\item its \emph{induced tree cone} $\text{cone} \{ \xi^W : \text{ $W$ is an induced tree of $G$}   \}$
		\item its \emph{induced tree polytope} $\text{conv} \{ \xi^W : \text{ $W$ is an induced tree of $G$}   \}$
		\item its \emph{induced path cone} $\text{cone} \{ \xi^W : \text{ $W$ is an induced path of $G$}\}$
		\item its \emph{induced path polytope} $\text{conv} \{\xi^W : \text{ $W$ is an induced path of $G$}\}$
	\end{itemize}

	\paragraph{Notations}
	Here and later, calligraphic letters denote polyhedra.
	Capital letters denote sets of vertices or edges, or graphs (note that $\delta$ is also used for sets of edges). Lowercase letters refer to vertices, edges or numbers.
	The symbol $\Gamma$ is reserved for sets of induced trees,
	while $\Psi$ for sets of induced paths, $\Omega$ for sets of cliques and $\lambda$ refers to vector of integer coefficients. Finally, $\Delta$ refers to gaps which we aim to reduce to 0.

	\section{Induced tree polytope of chordal graphs}

	In this section, we show that the induced tree's extended incidence vectors of chordal graphs form a Hilbert basis, and then we use this fact to exhibit a linear system describing the induced tree polytope of chordal graphs.
	
	\subsection{Induced trees of chordal graphs and Hilbert basis}
	
	Given a graph $G$, and a vertex $v$ we denote by $\Omega^v$ the set of maximal cliques containing $v$ and by $\mathcal{C}_T(G)$ the cone defined by the following system of inequalities:
	
	\begin{align}
		\sum_{u \in K\setminus \{v\}} y_{uv} &\le x_v  && \forall v \in V, \ K \in \Omega^v\label{iq:neighClique}\\
		-x_v &\le 0 && \forall v \in V\\
		-y_e &\le 0 && \forall e \in E \label{iq:trivialYe}
	\end{align}
	
	For binary points, Inequalities~\eqref{iq:neighClique} express the fact that no two edges incident to $v$ and contained in a same triangle belong to an induced tree.
	Observe that induced trees’ extended incidence vectors belong to $\mathcal{C}_T(G)$.
	
	\begin{theorem}\label{the:Hilbertbasistree}
		The induced tree extended incidence vectors of a chordal graph $G$ form a Hilbert basis of $\mathcal{C}_T(G)$.
	\end{theorem}
	
	\begin{proof}
		We proceed by induction. For initialization, consider the graph $G$ with a single vertex. $\mathcal{C}_T(G)$ is the positive half-line trivially generated by the unit vector.
		
		Induction goes as follows: let $v$ be a simplicial vertex of $G$ and suppose that the integer points of $\mathcal{C}_T(G\setminus v)$ are obtained as positive integer combinations of induced trees' extended incidence vectors. 
		Let $(x^*,y^*)$ be an integer point of $\mathcal{C}_T(G)$, and $(\lambda,\Gamma)$ an integer nonnegative combination of $G$'s induced trees' extended incidence vectors satisfying the following properties, where $(x',y') = \sum_{T \in \Gamma} \lambda_T \Gamma_T$:
	    \begin{enumerate}[label=(\roman*)]
			\item \label{condd:1} $(x',y') \le (x^*,y^*)$
			\item \label{condd:2} $ (x',y')|_{G\setminus v} = (x^*,y^*)|_{G\setminus v} $ 
			\item \label{condd:3} $(\lambda,\Gamma)$ minimize the quantity $\begin{array}{l}
				\\
				\Delta(\lambda,\Gamma) := \sum_{e \in \delta(v)} \Delta_e(\lambda, \Gamma)\\
				\text{where } \Delta_e(\lambda, \Gamma) := y_{e}^* - y_e' \ \forall e \in \delta(v) \end{array}$ among the combinations satisfying~\ref{condd:1}-\ref{condd:2}.
		\end{enumerate}
		$\Delta_{uv}(\lambda,\Gamma)$ corresponds to a Manhattan distance to $(x^*,y^*)$. 
		An integer nonnegative combination $(\lambda,\Gamma)$ satisfying Conditions~\ref{condd:1}-\ref{condd:2} is given by the induction hypothesis applied on the integer point $(x^*,y^*)|_{G\setminus v}$ of $\mathcal{C}_T(G\setminus v)$ hence, $(\lambda,\Gamma)$ exists.
		
		\medskip
		
		We first show that $\Delta_{uv}(\lambda,\Gamma) = 0 \ \forall uv \in E$.
		Fix $uv \in E$ and $H = G\setminus uv$, let us consider the following disjoints subsets of $\Gamma$:
		\begin{itemize}
			\item $\Gamma^u \ \forall u \in N(v)$ containing its elements whose intersection with $N[v]$ equals $\{u\}$.
			\item $\Gamma^u_2 \ \forall u \in N(v)$ containing its elements whose intersection with $N(v)$ is $\{u,w\}$ for some $w \in N(v)$.
			\item $\Gamma^v_u \ \forall u \in N(v)$ containing the elements having $v$ as an extremity and containing $u$.
		\end{itemize}  
		Note that $u$ is either a leaf or an internal node in the elements of $\Gamma^u_2$. We refer to Figure~\ref{fig:proofinducedtree} for an illustration of the different types of induced trees considered, it may be used as a visual aid to keep track of the definitions, we provide two different corresponding induced trees in different shades of green; these are the two bottom ones.

		\medskip

		Suppose by contradiction that $\Delta_{uv}(\lambda,\Gamma)$ is nonzero and that $\Gamma^u$ is nonempty for some neighbor $u$ of $v$. 
		Let $P$ be an element of $\Gamma^u$ and set $\alpha = \min (\Delta_{uv}(\lambda,\Gamma), \lambda_P)$.
		Note that $P \cup \{v\}$ is an induced tree, therefore decreasing $\lambda_P$ and increasing $\lambda_{P \cup \{v\}}$ by $\alpha$ yields a nonnegative integer combination of induced paths' extended incidence vectors of $G$, say $(\lambda',\Gamma')$. As $\xi^{P}|_{G\setminus uv} = \xi^{P \cup \{v\}}|_{G\setminus uv} $, $(\lambda',\Gamma')$ satisfies Conditions~\ref{condd:1}-\ref{cond:2} and $\Delta_e(\lambda,\Gamma) = \Delta_{e}(\lambda',\Gamma') \ \forall e \in \delta(v) \setminus \{uv\}$. Moreover, $\Delta_{uv}(\lambda',\Gamma') = \Delta_{uv}(\lambda,\Gamma)-\alpha$, Condition~\ref{condd:3} applied to $(\lambda,\Gamma)$ implies that $\alpha = 0$. Therefore, either $\Gamma^u$ is empty or $\Delta_{uv}(\lambda,\Gamma) =0$. Suppose that $\Gamma^u$ is empty, we now show that $\Delta_{uv}(\lambda,\Gamma) =0$ also holds.
		
		Given a family $\Gamma'$ of induced trees, we denote by $\lambda(\Gamma')$ the sum of $\lambda$'s components associated with elements of $\Gamma'$.

		As $(x^*,y^*)$ belongs to $\mathcal{C}_T(G)$ we have the following:
		$$ y_{uv}^* + \lambda(\Gamma_{2}^u) = y^*_{uv}+ y^*(\delta(u,N(v) \setminus \{u\} )) \le x_u^* = x_u' = \lambda(\Gamma^v_u) + \lambda(\Gamma_2^u) =   y_{uv}' + \lambda(\Gamma_2^u)$$
		
		The first equality is given by the fact that the only elements of $\Gamma$ contributing to $y_{e}',$ with $e\in \delta(u,N(v) \setminus \{u\} )$ are the ones of $\Gamma^u_2$ moreover, each element of $\Gamma^u_2$ contributes to at most one $y_e$ for $e \in \delta(u,N(v) \setminus u)$.
		The inequality is given by inequality~\eqref{iq:neighClique} associated with vertex $u$ and clique $N[v]$ .
		The following equality holds by Condition~\ref{condd:2}.
		 The two last equalities hold by the facts that the elements of $\Gamma$ contributing to $x_u'$ are the ones in $\Gamma^u \cup \Gamma^u_2 \cup \Gamma^v_u$, but $\Gamma^u$ is supposed to be empty, and that the only elements of $\Gamma$ contributing to $y_{uv}'$ are in $\Gamma^v_u$.
		
		This shows that $\Delta_{uv}(\lambda,\Gamma) =0$ holds for all edges $uv \in \delta (v)$. Finally, the singleton $\{v\}$ can be added to $\Gamma$ and $\lambda_{v}$ set to $x_v^* - y(\delta(v))^*$ to achieve $(x',y') = (x^*,y^*)$ with nonnegative integer coefficients $\lambda$ on induced tree's extended incidence vectors.
		
	\end{proof}
	
	Note that Theorem~\ref{the:Hilbertbasistree} is not a characterization of the graphs for which the induced tree’s  extended incidence vectors form a Hilbert basis. In fact, this statement is also true for $C_4$. Moreover, it is straightforward to see that this property is closed under taking induced subgraphs. 
	 A natural question would be to give a forbidden induced subgraph characterization of the graphs for which this property holds.

	\subsection{Induced trees of chordal graphs and integral systems}

	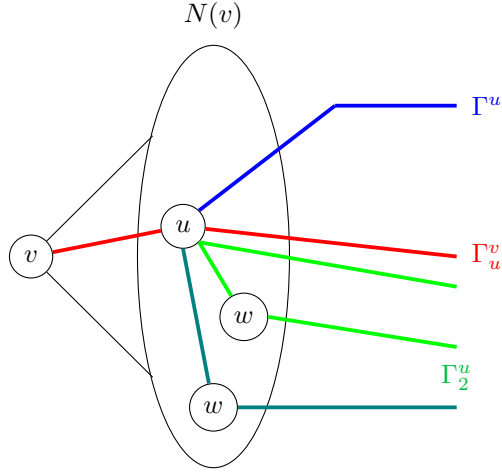
\begin{figure}
		\centering
		\begin{tikzpicture}
			[state/.style={circle, draw, minimum size=0.5cm},scale = 0.4]
			
			\node[state,draw,circle] (1)at(-2,3) {$v$};
			\node[state,draw,circle] (2)at(3,4) {$u$};
			\node at (4,11) {$N(v)$};
			\node[state,draw,circle] (3) at (4,-2) {$w$};
			\node[state,draw,circle] (4) at (5,1) {$w$};
			
			\node[color=blue,line width=0.05cm] at (13,8) {$\Gamma^u$};
			\node[color=green!50!teal,line width=0.05cm] at (12,-1) {$\Gamma^u_2$}; 
			\node[color=red,line width=0.05cm] at (13,3) {$\Gamma^v_u $ };
			
			\draw (4,3) ellipse (2.5cm and 7cm);
			\draw (1) -- (2,7) ;
			\draw (1) -- (2,-1) ;
			\draw[color=blue,line width=0.05cm] (2.north east) -- (8,8);
			\draw[color=blue,line width=0.05cm] (8,8) -- (12,8); 
			\draw[color=teal,line width=0.05cm] (2.south) -- (3);
			\draw[color=teal,line width=0.05cm] (3) -- (12,-2);
			\draw[color = green,line width=0.05cm] (4)-- (2.south east) ;
			\draw[color = green,line width=0.05cm](2.south east) -- (12,2) ;
			\draw[color = green,line width=0.05cm](4.east) -- (12,0);
			\draw[color = red,line width=0.05cm] (1) -- (2);
			\draw[color = red,line width=0.05cm] (2)-- (12,3);
			
		\end{tikzpicture}
		
		\caption{Sets of induced trees considered in the proof of Theorem~\ref{the:inducedPathChordal}}
		\label{fig:proofinducedtree} 
		
	\end{figure}
	
	Let us denote by $\mathcal{T}(G)$ the polyhedron defined as the intersection of $\mathcal{C}_T(G)$ and the following hyperplane:
	
	\begin{equation}
		\{ (x,y) : x(V) - y(E) = 1\} \label{iq:VboundE}
	\end{equation}

	\begin{theorem} \label{the:inducedPathChordal}
		A graph $G$ is chordal if and only if $\mathcal{T}(G)$ is binary.  
	\end{theorem}
	\begin{proof}
		$"\Rightarrow"$
		Suppose $G$ to be chordal.
		The extreme points of $\mathcal{T}(G)$ are exactly the scalings of $\mathcal{C}_T(G)$'s generators which belong to~\eqref {iq:VboundE}.
		As $\mathcal{C}_T(G)$ is defined with rational coefficients, each of its one-dimensional face contains integer points. This implies that a Hilbert basis of $\mathcal{C}_T(G)$ generates $\mathcal{C}_T(G)$. Therefore, Theorem~\ref{the:Hilbertbasistree} gives that induced trees extended incidence vectors of a chordal graph $G=(V,E)$ generate $\mathcal{C}_T(G)$. Extended incidence vectors of induced trees belong to~\eqref{iq:VboundE} and hence are extreme points of $\mathcal{T}(G)$. It remains to show that $\mathcal{T}(G)$ is bounded.  For each vertex $v$ of $G$, there exists a perfect elimination ordering ending in $v$, say $(v_1,\dots,v_n,v)$, this is a consequence of the fact that minimal separators are cliques~\cite{Dirac1961}. By definition of perfect elimination orderings, $N(v_i) \cap \{ v_{i+1},\dots,v_n, v\}$ is a clique and the following inequalities are valid for $\mathcal{C}_T(G)$ : 
		$$ y(\delta(v_i,\{v_{i+1},\dots,v_n,v\} )) \le x_{v_i} \quad \forall  i \in \{1,\dots, n\}.$$
		Therefore, summing the previous family of inequalities with  $x(V) - y(E) = 1$ yields $x_v \le 1$. Moreover, $y$ variables are sandwiched between 0 and $x$ variables.

		$"\Leftarrow"$
		By contradiction, suppose that $G$ is not chordal, and let $H$ induce a hole of length at least 4 in $G$, and let $v$ be a vertex of $H$. By setting $x_w = 1 \ \forall w \in H \setminus v$, $y_e = 1 \ \forall e \in E(H)$, $x_v = 2$ and 0 on all other components, we obtain a point $(x,y)$ which belongs to $\mathcal{T}(G)$ and is not included in the hypercube. This implies that $\mathcal{T}(G)$ is itself not contained in the hypercube and hence not binary.
	\end{proof}

	Since the set of inclusion-wise maximal cliques of chordal graphs is polynomial then, \eqref{iq:neighClique}--\eqref{iq:trivialYe} has compact size.
	Therefore, optimizing a linear function on the vertex and edge set over $\mathcal{T}(G)$ is doable in polynomial time. We denote by $\mathcal{K}_G$ the set of maximal cliques of $G$.
	
	\begin{corollary}\label{cor:inducedtreechordal}
		Given a graph $G$, the following are equivalent:
		\begin{itemize}
			\item $G$ is chordal
			\item the induced tree polytope of $G$ is equal to $\mathcal{T}(G)$
			\item $\mathcal{T}(G)$ is bounded. 
				\end{itemize}
		 Moreover, when $G$ is chordal, $\mathcal{T}(G)$ is described by a linear system of size $O(|V||\mathcal{K}_G|)$.
	\end{corollary}

	\begin{corollary}\label{cor:inducedpathchordal}
		The edge and vertex weighted maximum induced tree is solvable in polynomial time on chordal graphs.
	\end{corollary}

	\section{Induced path polytope of chordal graphs} \label{sec:inducedpath}
	
	We first recall that chordal graphs are closed by edge contraction as exhibited in~\cite{dupontbouillard2024contractions} and use this result to prove a technical lemma.
	Our proof of the fact that induced paths' extended incidence vectors of chordal graphs form Hilbert bases follows a similar scheme to the one of the previous section. 
	
	We define $(\lambda,\Psi)$ in a similar way to what we did with $(\lambda,\Gamma)$ in the proof of Theorem~\ref{the:Hilbertbasistree}. However, the proof of the analogous result for induced path requires a closer study of $(\lambda,\Psi)$. Therefore, we first exhibit a set of properties of $(\lambda,\Psi)$ in the form of lemmas before diving into the proof itself. 
	
 	In the second section, we deduce an exponential size description of the induced path polytope of chordal graphs and characterize the inequalities that define facets. Finally, we show that the non-dominated constraints come in a polynomial number.

	\begin{lemma}[\cite{dupontbouillard2024contractions}]\label{lemma:chordalcontract}
		If $G=(V,E)$ is chordal, then $G/F$ is also chordal for any set $F \subseteq E$.
	\end{lemma}
	
	The proof that chordal graphs' induced path incidence vectors form a Hilbert basis requires to "reorganize" the induced path of a conic combination. We mean by reorganize that we build a new induced path having desirable properties from other induced path. This is what stands for the following lemma.

	\begin{claim}\label{claim:adjacentorinducepaths}
		Given $P_1, P_2$ be two induced paths of a chordal graph having the same extremity $u$, and such that the neighbor $u_1$ of $u$ in $P_1$ is different from the neighbor $u_2$ of $u$ in $P_2$.  Either $u_1$ is adjacent to $u_2$ or $P_2 \cup P_1$ induces a path.
		
	\end{claim}
	\begin{proof}
		The result is trivial if $P_1$ contains at most two vertices. Suppose now that $P_1$ contains at least 3 vertices.
		By contradiction, suppose that $P_2 \cup P_1$ does not induce a path and that $u_1$ is not adjacent to $u_2$.
		Let $w_1$ be the other element of $P_1$ adjacent to $u_1$, by hypothesis, $w_1 \neq u_2$.
		
		As $P_1$ and $P_2$ have an extremity in common but their union does not induce a path, either their intersection contains $u$ and another vertex, or there exists an edge going from $P_1 \setminus u$ to $P_2\setminus u$. 
		In both situations, either $u_2$ is adjacent to $u_1$ or
		there exist a cycle with vertex set $C\subseteq P_1 \cup P_2¨$ containing $\{u,u_1,w_1,u_2\}$.
		Let us contract all the edges in $E(P_1  \cup P_2  \setminus \{u_1,u_2,u\})$ to a single vertex $w$. The subgraph induced by  $\{u_1,u_2,u,w\}$ contains a cycle of length 4, and $w$ is non-adjacent to $u$ as otherwise $P_1$ and $P_2$ would not induce paths. By chordality of $G$ and Lemma~\ref{lemma:chordalcontract}, $u_2$ is adjacent to $u_1$, a contradiction.
		
	\end{proof}

To show that a set of vectors forms a Hilbert basis, it can be easier to know the polyhedral description of the cone they generate. Therefore, we introduce $\mathcal{C}_P(G)$ as the cone defined by the following linear system.

\begin{align}
	y\Biggl(\delta\Bigl(w,C(K\cup \{w\})\Bigr)\Biggr) + 2y\Biggl(\delta \Bigl(w, K \Bigr)\Biggr) &\le 2x_w &&  \ \forall w \in V, \  \text{ clique } K \subseteq N(w) \label{iq:neighCliquebb}\\
	-x_w & \le  0 &&\forall w \in V \label{iq:trivialx}\\
	-y_e &\le 0 && \forall e \in E \label{iq:trivialYebb}  
\end{align}


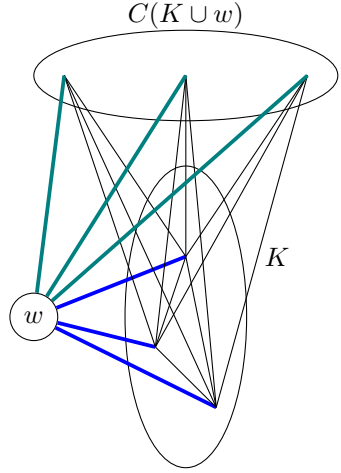
\begin{figure}[h!]
	\centering
	\begin{tikzpicture}
		[state/.style={circle, draw, minimum size=0.5cm},scale = 0.4]
		\node at (8,2){$K$};
		\node[state,draw,circle] (1)at(0,0) {$w$};
		\node at (5,10) {$C(K\cup {w})$};
		\draw (5,0) ellipse (2cm and 5cm);
		\draw[color=blue,line width=0.05cm] (1) -- (4,-1);
		\draw[color=blue,line width=0.05cm] (1) -- (5,2);
		\draw[color=blue,line width=0.05cm] (1) -- (6,-3);
		\draw (4,-1) -- (5,2);
		\draw (5,2) -- (6,-3);
		\draw (6,-3) -- (4,-1);
		\draw (5,8) ellipse (5cm and 1.5cm);
		\draw (1,8) -- (5,2);
		\draw (1,8) -- (6,-3);
		\draw (1,8) -- (4,-1);
		\draw (5,8) -- (5,2);
		\draw (5,8)-- (6,-3);
		\draw (5,8) -- (4,-1);
		\draw (9,8) -- (5,2);
		\draw (9,8)-- (6,-3);
		\draw (9,8) -- (4,-1);
		\draw[color=teal,line width=0.05cm] (1) -- (1,8);
		\draw[color=teal,line width=0.05cm] (1) -- (5,8);
		\draw[color=teal,line width=0.05cm] (1) -- (9,8); 
	\end{tikzpicture}
	
	\caption{Structure of inequalities~\eqref{iq:neighCliquebb} for a given vertex $w$ and clique $K$ of its neighborhood.}
	\label{fig:idneighcliquebb}
	
\end{figure}

Note that Inequalities~\eqref{iq:neighCliquebb} are exactly Inequalities~\eqref{iq:neighClique} when $K$ is a maximal clique of $N(w)$ which implies the following.

\begin{obs}\label{obs:contained}
Given a graph $G$, $\mathcal{C}_P(G)$ is included in $\mathcal{C}_T(G)$.
\end{obs}

For binary points, Inequalities~\eqref{iq:neighCliquebb} express the fact that an induced path contains at most two edges in $\delta(w)$ (for $K = \emptyset$), and the extremities of two such edges cannot induce a triangle (for $K \neq \emptyset$).
We refer to Figure~\ref{fig:idneighcliquebb} for an illustration where the blue edges are the ones involved with coefficient 2 and the green ones are involved with coefficient 1. At most two green edges belong to a same induced path; given a blue edge $e$, for any other colorful edge $e'$, there exists a triangle containing both, therefore no induced path $P$ verifies $\{e,e'\} \subseteq E(P)$. This shows that extended incidence vectors of induced paths belong to $\mathcal{C}_P(G)$.

\begin{obs}\label{obs:belongcpg}
	The induced paths' extended incidence vectors of a graph $G$ belong to $\mathcal{C}_P(G)$.
\end{obs}

	\subsection{Induced paths of chordal graphs and Hilbert basis} \label{sec:inducedpath1}

	This section is dedicated to the proof of the following theorem.
	\begin{restatable}{theorem}{mainthm}\label{the:hbPath}
		The induced path extended incidence vectors of a chordal graph $G$ form a Hilbert basis of $\mathcal{C}_P(G)$.
	\end{restatable}

	We prove Theorem~\ref{the:hbPath} by induction over the vertex set of any chordal graph $G$. Note that the initilization is given by the single vertex graph for which $\mathcal{C}_P(G)$ is the positive half line trivially generated by a unit vector. The induction proof consists in considering an integer point $(x^*,y^*)$ in $\mathcal{C}_P(G)$, a simplicial vertex $v$ of $G$ (which exists by chordality of $G$) and a nonnegative integer combination of $G$'s induced paths' extended vectors say, the induced path $\Psi$ with coefficient $\lambda$ satisfying the following properties where $(x',y') = \sum_{P \in \Psi} \lambda_P \xi^P$:
	    \begin{enumerate}[label=(\roman*)]
		\item \label{cond:1} $(x',y') \le (x^*,y^*)$
		\item \label{cond:2} $ (x',y')|_{G\setminus v} = (x^*,y^*)|_{G\setminus v} $ 
		\item \label{cond:3} $(\lambda,\Psi)$ minimize the quantity $\begin{array}{l}
			\\
			\Delta(\lambda,\Psi) := \sum_{e \in \delta(v)} \Delta_e(\lambda, \Psi)\\
			\text{where } \Delta_e(\lambda, \Psi) := y_{e}^* - y_e' \ \forall e \in \delta(v) \end{array}$ among the ones that satisfy Conditions~\ref{cond:1}-\ref{cond:2}.
	\end{enumerate}

	Note that $(x^*,y^*)|_{G \setminus v}$ belong to $\mathcal{C}_P(G\setminus v)$, therefore the induction hypothesis implies that $(x^*,y^*)|_{G \setminus v}$ is obtained as an integer conic combination of induced paths of $G\setminus v$. 
	This combination, say $(\Psi',\lambda')$, satisfies Conditions~\ref{cond:1}-\ref{cond:2} but is somehow the worst in terms of Condition~\ref{cond:3} as $\Delta(\lambda',\Psi') = y^*(\delta(v))$. Nevertheless, this proves that $(\lambda,\Psi)$ exists.	
	For the sake of simplicity, we consider that when $\lambda_P = 0$, $P$ does not belong to $\Psi$. 
	
	\medskip

	We first propose a decomposition of $\Psi$, and in the form of lemmas, we exhibit some properties satisfied by $(\lambda,\Psi)$.
	Given $u \in N(v)$,  let us decompose $\Psi$ into the following subsets: 
		\begin{itemize}
		\item  $\Psi^u$ is the subset of $\Psi$ whose elements have $u$ as an extremity and do not contain any other element of $N[v]$ 
		\item  $\Psi_2^u $ is the subset of $\Psi$ whose elements have $u$ as an extremity and contain another vertex $w$ in $N(v)$
		\item $\Psi_{I_2}^u$ is the subset of $\Psi$ whose element's interior contain $u$ and another element $w$ of $N(v)$.
		\item $\Psi_{I_1}^u$ is the subset of $\Psi$ whose elements's interior contain $u$ but no other vertex of $N(v)$
		\item $\Psi_u^v $ is the subset of $\Psi$ whose elements contain both $u$ and $v$, note that $v$ is necessarily an extremity of $\Psi_u^v$'s elements
		\item $\Psi_\bot$ contains the remaining elements.
	\end{itemize}

\begin{figure}[h]
	\centering
	\begin{tikzpicture}
		[state/.style={circle, draw, minimum size=0.5cm},scale = 0.4]
		
		\node[state,draw,circle] (1)at(-2,3) {$v$};
		\node[state,draw,circle] (2)at(3,4) {$u$};
		\node at (4,11) {$N(v)$};
		\node[state,draw,circle] (3) at (4,-2) {$w$};
		\node[state,draw,circle] (4) at (5,1) {$w$};
		
		\node[color=blue] at (13,10) {$\Psi^u$};
		\node[color=teal] at (13,-2) {$\Psi^u_2$}; 
		\node[color=red] at (13,7) {$\Psi^v_u $ };
		\node[color=green] at (13,0) {$\Psi_{I_2}^u$};
		\node[color = violet] at (13,4) {$\Psi_{I_1}^u$};
		\draw (4,3) ellipse (2.5cm and 7cm);
		\draw (1) -- (2,7) ;
		\draw (1) -- (2,-1) ;
		\draw[color=blue,line width=0.05cm] (2.north) -- (8,10);
		\draw[color=blue,line width=0.05cm] (8,10) -- (12,10); 
		\draw[color=teal,line width=0.05cm] (2.south) -- (3);
		\draw[color=teal,line width=0.05cm] (3) -- (12,-2);
		\draw[color = green,line width=0.05cm] (4)-- (2.south east) ;
		\draw[color = green,line width=0.05cm](2.south east) -- (12,1) ;
		\draw[color = green,line width=0.05cm](4.east) -- (12,-1);
		\draw[color = red,line width=0.05cm] (1) -- (2);
		\draw[color = red,line width=0.05cm] (2)-- (12,7);
		\draw[color = violet,line width=0.05cm] (2) -- (12,5);
		\draw[color = violet,line width=0.05cm] (2) -- (12,3);
		
	\end{tikzpicture}
	\caption{Types of induced paths considered in Section~\ref{sec:inducedpath1}}
	\label{fig:inducedpathgraph}
\end{figure}
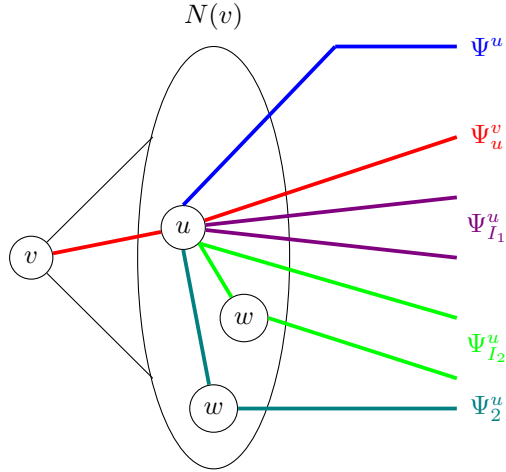

We denote by $K_2^u$ the subset of $N(v)\setminus \{u\}$ contained in at least one path in $\Psi_2^u$, that is:
$$K_{2}^u := \{w \in N(v)\setminus \{u\} : \exists P \in \Psi_2^u  \text{ s.t. } w \in P \}. $$
	As the only induced path of $\Psi$ containing the edge $uv$ belong to $\Psi^v_u$ the following holds:
\begin{equation}
	y_{uv}' = \lambda(\Psi_u^v). \label{eq:decompyuv}
\end{equation}
We refer to Figure~\ref{fig:inducedpathgraph} for an illustration of the different types of path considered, it may be used as a visual aid to keep track of the definitions.

	\begin{lemma}\label{lem:psiuemptyR}
		Given $u \in N(v)$, either $\Delta_{uv}(\Psi,\lambda) = 0$ or $\Psi^u = \emptyset$.
	\end{lemma}
	\begin{proof}
		By contradiction let $P$ be an element of $\Psi^u$ and $\alpha=  \min \Bigl(\Delta_{uv}(\lambda,\Psi ), \lambda_P \Bigr)$. Decreasing $\lambda_P$ and increasing $\lambda_{P \cup \{v\}}$ by $\alpha$ yields a new integer conic combination $(\lambda',\Psi')$, of induced paths' extended incidence vectors. Moreover, as $\xi^{P}|_{G\setminus uv} = \xi^{P \cup v}|_{G\setminus uv}$, we have that  $\Bigl( \sum_{P \in \Psi'}\lambda_P' \xi^{P} \Bigr)\Bigm\lvert_{G\setminus uv}  = (x',y')|_{G \setminus uv}$. Therefore $(x',y')$ satisfies Conditions~\ref{cond:1}-\ref{cond:2}, $\Delta_e(\lambda',\Psi') = \Delta_e(\lambda,\Psi) \ \forall e \in \delta(v) \setminus uv$ and $\Delta(\lambda',\Psi') = \Delta(\lambda,\Psi) - \alpha$. By Condition~\ref{cond:3} applied to $(\lambda,\Psi)$ we obtain $\alpha = 0$ and as $\lambda_P > 0 $ we finally obtain $\Delta_{uv}(\Psi,\lambda) = 0$. 
	\end{proof}

	By Lemma~\ref{lem:psiuemptyR}, if $\Delta_{uv}(\Psi,\lambda)$ is nonzero, we have $\Psi^u = \emptyset$. Therefore, the following inequality holds and will be used in the upcoming results:
	\begin{equation}
		x^*_u  = x_u' =  \lambda(\Psi_2^u) +\lambda(\Psi_{I_2}^u) + \lambda(\Psi_{I_1}^u) + \lambda(\Psi_u^v). \label{eq:decompxu} \\
	\end{equation}

	\begin{lemma}\label{lem:psiuI1nonempty}
		Given $u \in N(v)$, either $\Delta_{uv}(\lambda,\Psi) = 0$ or $\Psi^u_{I_1}$ is nonempty.
	\end{lemma}
	\begin{proof}
		Suppose that $\Delta_{uv}(\lambda,\Psi) \neq 0$ and that $\Psi^u_{I_1} = \emptyset$. 
		As $(x^*,y^*)$ belongs to $\mathcal{C}_P(G)$, by inequality~\eqref{iq:neighCliquebb} associated with $u$ and $K= N[v]\setminus \{u\}$ we get: $x_u^* \ge y^*(\delta(u,K))$ and the right-hand side is equal to $ \lambda(\Psi_2^u) + \lambda(\Psi_{I_2}^u) +y_{uv}^*$ by definition of $(\lambda,\Psi)$. Using equalities~\eqref{eq:decompyuv} and~\eqref{eq:decompxu} we get : $ y_{uv}' \ge y_{uv}^*$ and hence $\Delta_{uv}(\lambda,\Psi) = 0$, a contradiction.
	\end{proof}
	
	\begin{lemma}\label{lem:psiuI1complete}
		Given $u \in N(v)$, the neighbors of $u$ in the elements of $\Psi_{I_1}^u$ are complete to $K_{2}^u \cup \{u\}$.
	\end{lemma}
\begin{proof}
	By contradiction, suppose that a neighbor $u_1$ of $u$ in some element $P \in \Psi^u_{I_1}$ is not adjacent to a vertex $u_2 \in K_2^u$ contained in some induced path $P_2 \in \Psi^u_2$.  Let $\alpha = \min \Bigl(\Delta_{uv}(\lambda,\Psi), \lambda_{P}, \lambda_{P_2} \Bigr)$. 
	Let us consider the two inclusion wise maximal induced path contained in $P$ ending in $u$, say $P_1$, $P'$, and let $P_1$ contain $u_1$. By Claim~\ref{claim:adjacentorinducepaths}, $P_1 \cup P_2$ induces a path of $G$.
	Decreasing $\lambda_{P}$, $\lambda_{P_2}$ and increasing $\lambda_{P_1 \cup P_2}$, $\lambda_{P' \cup \{v\}}$ yields a nonnegative integer combination $(\lambda', \Psi')$ of induced paths' extended incidence vectors. Moreover, the following equality holds $\Bigl(\xi^{P} + \xi^{P_2}\Bigr)\Bigm\lvert_{G\setminus uv} = \Bigl(\xi^{P_1 \cup P_2} + \lambda_{P' \cup \{v\}}\Bigr)\Bigm\lvert_{G\setminus uv}$. Therefore, $(\lambda',\Psi')$ satisfies Conditions~\ref{cond:1}-\ref{cond:2}, $\Delta_e(\lambda',\Psi') = \Delta_e(\lambda,\Psi) \ \forall e \in \delta(v) \setminus \{uv\}$ and  $\Delta_{uv}(\lambda',\Psi') =\Delta_{uv}(\lambda,\Psi) - \alpha $. We finally get $\alpha=  0$ and hence $\Delta_{uv}(\lambda,\Psi) = 0$. 
\end{proof}

\begin{lemma}\label{lem:lambdauvzero}
	Given $u \in N(v)$ either $\Delta_{uv}(\lambda,\Psi) = 0$ or $\lambda_{u,v} = 0$.
\end{lemma}
\begin{proof}
	Suppose by contradiction that both $\lambda_{u,v} $ and $\Delta_{uv}(\lambda,\Psi)$ are nonzero. By Lemma~\ref{lem:psiuI1nonempty} $\Psi_{I_1}^u$ is nonempty,	
	let $P$ be an element of $\Psi_{I_1}^u$. $P$ contains the two distinct inclusion-wise maximal subsets  $P_1$ and $P_2$ both inducing paths ending in $u$. 
	Let $\alpha = \min (\Delta_{uv}(\lambda,\Psi),\lambda_P,\lambda_{\{u,v\}})$. Note that $P_1 \cup \{v\}$ induces a path. 
	Decreasing $\lambda_{u,v}$, $\lambda_P$ and increasing  $\lambda_{P_2} $, $\lambda_{P_1 \cup \{v\}} $ 
	by $\alpha$ yields a nonnegative integer combination $(\lambda',\Psi')$. As $\xi^{\{u,v\}} + \xi^P = \xi^{\{u,v\} \cup P_1} + \xi^{P_2} $, $(\lambda',\Psi')$ satisfies Conditions~\ref{cond:1}-\ref{cond:2} and $\Delta_{uv}(\lambda,\Psi) = \Delta_{uv}(\lambda',\Psi')$. However, $(\Psi')^u$ is nonempty as it contains $P_2$, by Lemma~\ref{lem:psiuemptyR}, $\Delta_{uv}(\lambda,\Psi) = 0$, a contradiction.
\end{proof}

\begin{lemma}
	Given $u \in N(v)$, either $\Delta_{uv}(\lambda,\Psi) = 0$ or $K_2^u$ is nonempty.
\end{lemma}	
\begin{proof}
	Suppose by contradiction that $\Psi_2^u$ is empty and that $\Delta_{uv}(\lambda,\Psi)$ is strictly positive. 
	By Lemma~\ref{lem:psiuemptyR}, Equality~\eqref{eq:decompxu} holds. Inequality~\eqref{iq:neighCliquebb} associated with $u$ and the empty clique applied to $(x^*,y^*)$ yields: $2x_u^* \ge y^*(\delta(u))$.
	By Lemma~\ref{lem:lambdauvzero}, $\lambda_{u,v}=0$. Therefore, each $\lambda_P$ for $P \in \Psi_{I_1}^u \cup \Psi_{I_2} \cup \Psi^v_u $ contributes two times in $y^*(\delta(u))$, hence the following holds:
	$$  2x_u^* \ge 2\Bigl( \lambda(\Psi_{I_1}^u) + \lambda(\Psi_{I_2}^u)\Bigr) + \lambda(\Psi_u^v) + y_{uv}^*. $$
	Equalities~\eqref{eq:decompyuv}-\eqref{eq:decompxu} yields $y_{uv}' \ge y_{uv}^*$ implying that $\Delta_{uv}(\lambda,\Psi) = 0$, a contradiction.	
\end{proof}

\begin{lemma}\label{lem:psiuvcomplete}
	Given $u\in N(v)$, either $\Delta_{uv}(\lambda,\Psi) = 0$ or the neighbors of $u$ in the elements of $\Psi_u^v$ are complete to $K_{2}^u \cup \{u\}$.
\end{lemma}
\begin{proof}
 Suppose that $\Delta_{uv}(\lambda,\Psi)$ is nonzero. By Lemma~\ref{lem:lambdauvzero}, the elements of $\Psi_u^v$ have at least 3 vertices. Let $P = (v,u,u_1,\dots)$ be an element of $\Psi^v_u$. By construction, $v$ is complete to $K_2^u \cup \{u\}$. Suppose that $u_1$ is not complete to $K_2^u \cup \{u\}$ and, without loss of generality, not adjacent to the vertex $u_2 \in P_2 \cap K_2^u$ for some induced path $P_2 \in \Psi_2^u$. Let 
 $P_1$ be the inclusion-wise maximal induced path included in $P$ ending in $u$ and containing $u_1$. By Claim~\ref{claim:adjacentorinducepaths}, $P_1 \cup P_2$ induces a path. 
 Let $\alpha = \min(\lambda_P, \lambda_{P_2}) $, decreasing $\lambda_{P}$, $\lambda_{P_2}$ and increasing $\lambda_{P_1 \cup P_2} $,  $\lambda_{\{u,v\}} $ by $\alpha$ yields a nonnegative integer combination $(\lambda',\Psi')$ of induced paths' extended incidence vectors. As $\xi^P+ \xi^{P_2} =  \xi^{P_1 \cup P_2} + \xi^{\{u,v\}}$,  $(\lambda',\Psi')$ satisfies Conditions~\ref{cond:1}-\ref{cond:3},  and $\Delta_{uv}(\lambda,\Psi) = \Delta_{uv}(\lambda',\Psi')$. Moreover, we have that $\lambda'_{\{u,v\}}$ is nonzero, therefore by Lemma~\ref{lem:lambdauvzero}, $\Delta_{uv}(\lambda,\Psi) = 0$.
 \end{proof}

	\mainthm*
	
	\begin{proof} 
		Similarly to the proof of Theorem~\ref{the:Hilbertbasistree}, the first step is to show that $\Delta_{e}(\lambda,\Psi)$ is equal to 0 for every $e \in \delta(v)$ for $(\lambda,\Psi)$ satisfying Conditions~\ref{cond:1}-\ref{cond:3}.
		This is done through the use of an inequality~\eqref{iq:neighCliquebb} associated with $u$. 
		Therefore, we need to exhibit the right clique of $N(u)$. 		
		We iteratively construct this clique that we denote by $K_{\Psi}$, and note $\Psi_{K_{\Psi}}$ the subset of $\Psi^u_2 \cup \Psi_{I_2}^u$ whose elements contain a vertex of $K_{\Psi}$. Note that $\Psi_{K_{\Psi}}$ is supposed to be changing along to $K_{\Psi}$.
	
		\medskip 
		Start with $K_{\Psi} = K_2^u$. If there exists an element $P$ of $\Psi_{I_2}^u \setminus \Psi_{K_{\Psi}}$ such that the vertex $ w \in \Bigl(P\cap N(u) \Bigr) \setminus N(v)$ is not complete to $K_{\Psi}$, add $P$ to $\Psi_{K_{\Psi}}$ and add the other neighbor of $u$ in $P$ to $K_{\Psi}$. This procedure is applied until the neighbors of $u$ in the elements of $\Psi_{I_2}^u \setminus \Psi_{K_{\Psi}}$ are complete to $K_{\Psi}\cup \{u\}$.

		\medskip
		
		By Lemmas~\ref{lem:psiuI1complete}-\ref{lem:psiuvcomplete}, the neighbors of $u$ in the elements of $ \Psi_{I_1}^u \cup \Psi_u^v$ are complete to $K_{2}^u \cup \{u\}$. Moreover, by the previous construction, the elements of $\Psi^u_{I_2} \setminus \Psi_{K_{\Psi}}$ satisfy this property for the whole clique $K_{\Psi} \cup \{u\}$.
		
		It remains to show that the neighbors of $u$ in the elements of $\Psi_{I_1}^u \cup \Psi_u^v$ are complete to $K_{\Psi} \setminus K_2^u$. We do so by reorganizing the elements of $\Psi_{K_{\Psi}}$.
		
		Each vertex of $K_{\Psi} \setminus K_2^u$ is contained in an element of $\Psi_{I_2}^{u} \cap \Psi_{K_{\Psi}}$. 
		Let $P$ be an element of $\Psi_{I_1}^u \cup \Psi^v_u$ and let $P'$ be an element of  $\Psi_{K_{\Psi}} \setminus \Psi_2^u$. 
		As $P'$ has been added to $\Psi_{K_{\Psi}}$ by the previous construction, there exists a sequence of induced path, say  $(P_1,\dots, P_{\ell})$ such that $P_1 = P'$, $P_{\ell} \in \Psi_{2}^u$ and all the other elements belong to $\Psi_{I_2}^u$, such that the vertex $u_i \in (N(u) \cap P_i) \setminus K_{\Psi}$ is nonadjacent to the vertex $v_{i+1} \in K_{\Psi} \cap P_{i+1} \ \forall i \in \{1,\dots,\ell -1\}$. 
		For each $P_i, \ i\in \{1,\dots,\ell-1\}$, we consider the two maximal induced path $U_i, V_i$ included in $P_i$ and ending in $u$, with $U_i$ containing $u_i$ and $V_i$ containing $v_i$, we set $V_\ell = P_{\ell}$. Claim~\ref{claim:adjacentorinducepaths} implies that $U_i \cup V_{i+1} \ \forall i \in \{ 1,\dots,\ell-1\}$ induces a path, see Figure~\ref{fig:helpproofuivi} for an illustration where the dashed line represents a nonedge.
		Therefore, for $\alpha = min \{ \lambda_{P_i} : i \in \{ 1,\dots, \ell \}     \}$  decreasing $\lambda_{P_i} \ \forall i \in \{1,\dots,\ell\}$, increasing
		$\lambda_{V_{i+1} \cup U_i} \forall i \in \{1,\dots,\ell\-1\}$ and $\lambda_{V_1},  $ by $\alpha$
		yields a nonnegative combination of induced paths' extended incidence vectors $(\lambda',\Psi')$.
		Moreover $\sum_{i = 1}^\ell \xi^{P_i} = \xi^{V_1} + \sum_{i = 1}^{\ell-1} \xi^{U_i \cup V_{i+1}}$,
		implying that, $(\lambda',\Psi')$ satisfies Conditions~\ref{cond:1}-\ref{cond:3} and $(\Psi')^u_2$ contains $V_1$. 
		By Lemmas~\ref{lem:psiuI1complete} and~\ref{lem:psiuvcomplete}, $P$ contains two vertices adjacent to the only vertex in $V_1\cap K_{\Psi}$.
		Note that this observation can be made for each element in $\Psi_{K_{\Psi}} \setminus \Psi_2^u$ and each element in $\Psi^u_{I_1} \cup \Psi^v_u$. This implies that the neighbors of $u$ in the elements of $\Psi^u_{I_1} \cup \Psi^v_u$ have two vertices complete to $K_{\Psi} \cup \{u\}$.	
		
		\medskip

		As $(x^*,y^*)$ is a point of $\mathcal{C}_P(G)$, it satisfies the constraint~\eqref{iq:neighCliquebb} associated with $K_{\Psi}$ and $u$ that is:
		\begin{equation}
			2x_u' = 2x_u^* \ge  y^*\Biggl(\delta \Bigl(u,C\left(K_{\Psi}\cup \{u\}\right)\Bigl)\Biggr) + 2y^*\Biggl(\delta \Bigl(u,K_{\Psi} \Bigr)\Biggr)  \label{iq:respectUKW}
		\end{equation} 
		By the previous paragraph, each element $P$ of $\Psi_u^v \cup \Psi_{I_1}^u \cup  \Psi_{I_2}^u \setminus \Psi_{K_{\Psi}}$ contains $u$ and its neighbors in $P$ are complete to $K_{\Psi}$.
		This implies that $\lambda_{P}$ contributes two times in $$y^*\Biggl(\delta \Bigl(u, C(K_{\Psi} \cup \{u\}) \Bigr)\Biggr).$$
		
		Therefore, the following holds:
		$$\begin{array}{rl}
			y^*\Biggl(\delta \Bigl(u,C(K_{\Psi} \cup \{u\})\Bigr)\Biggr) &= y^*_{uv} + y^*\Biggl(\delta \Bigl(u, C(K_{\Psi} \cup \{u\})\Bigr) \setminus \{uv\}\Biggr) \\ 
			&= y^*_{uv}+ \lambda(\Psi_u^v) + 2 \lambda(\Psi_{I_1}^u) + 2 \lambda(\Psi_{I_2}^u \setminus \Psi_{K_{\Psi}}).
		\end{array}$$ 
		Moreover, each element of $\Psi_{K_{\Psi}}$ intersects $K_{\Psi}$ and contains $u$, implying that it contributes one time in $y^*\Biggl(\delta \Bigl(u,K_{\Psi} \Bigr)\Biggr)$.
		Using equality~\eqref{eq:decompxu} and inequality~\eqref{iq:respectUKW} we get: 
		$$ 2\Bigl( \lambda(\Psi_2^u) +\lambda(\Psi_{I_2}^u)  + \lambda(\Psi_{I_1}^u) + \lambda(\Psi_u^v) \Bigr) \ge  y^*_{uv}+ \lambda(\Psi_u^v) + 2 \lambda(\Psi_{I_1}^u) + 2 \lambda(\Psi_{I_2}^u \setminus \Psi_{K_{\Psi}})+ 2 \Bigl(\lambda(\Psi_{K_{\Psi}})\Bigr) $$

		Moreover, $\Psi_{I_2}^u = (\Psi_{I_2}^u \cap \Psi_{K_{\Psi}} ) \cup (\Psi_{I_2}^u \setminus \Psi_{K_{\Psi}} )$ and $\Psi_{I_2}^u \cap \Psi_{K_{\Psi}} = \Psi_{K_{\Psi}} \setminus \Psi_{2}^u$.
		Therefore, $y_{uv}' \ge y_{uv}^*$ by equality~\eqref{eq:decompyuv}. 
		This reasoning can be made for each $uv \in \Delta(v)$. Therefore, $\Delta(\lambda,\Psi) = 0$. Setting $\lambda_{\{v\}} = (x^*_v - y^*(\delta(v)))$ for the trivial induced path $\{v\}$ yields the desired nonnegative integer combination of $(x^*,y^*)$ with induced paths' extended incidence vectors. 
\end{proof}

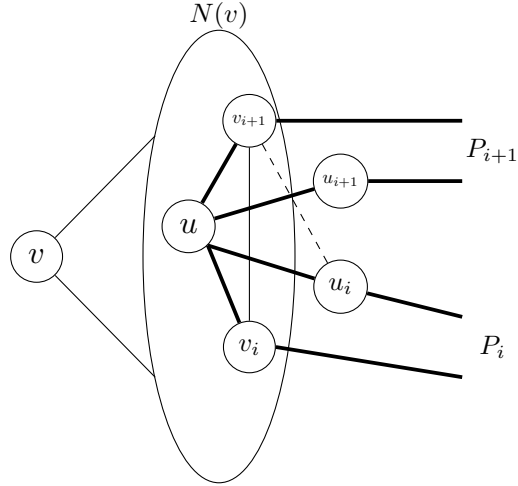
\begin{figure}
	\centering
	\begin{tikzpicture}
		[state/.style={circle, draw, minimum size=0.5cm},scale = 0.4]
		
		\node[state,draw,circle,scale = 1.2] (1)at(-2,3) {$v$};
		\node[state,draw,circle,scale = 1.2] (2)at(3,4) {$u$};
		
		\node at (4,11) {$N(v)$};
		\node[state,draw,circle] (4) at (5,0) {$v_i$};
		\node[state,draw,circle] (5) at (8,2) {$u_i$};
		
		\node[state,draw,circle,scale = 0.7] (6) at (5,7.5) {$v_{i+1}$};
		\node[state,draw,circle,scale = 0.7] (7) at (8,5.5) {$u_{i+1}$};

		\node at (13,0) {$P_i$};
\node at (13,6.5) {$P_{i+1}$};
		\draw (4,3) ellipse (2.5cm and 7.5cm);
		\draw (1) -- (1.9,7) ;
		\draw (1) -- (1.9,-1) ;
		\draw[dashed] (5) -- (6);
		\draw (4) -- (6);
				\draw [line width=0.05cm](2) -- (6) ;
		\draw [line width=0.05cm] (2) -- (7) ;
		\draw[line width=0.05cm] (6) -- (12,7.5);
		\draw[line width=0.05cm] (7)  -- (12,5.5); 
		\draw[line width=0.05cm] (4)-- (2.south east) ;
		\draw[line width=0.05cm](2.south east) -- (5) ;
		\draw [line width=0.05cm] (5) -- (12,1) ;
		\draw[line width=0.05cm](4.east) -- (12,-1);
		
	\end{tikzpicture}
\caption{Structure of the paths exhibited in Theorem~\ref{the:hbPath}}
\label{fig:helpproofuivi}	
\end{figure}

Remember that $C_4$ is a non-chordal graph for which the extended incidence vectors of induced trees form a Hilbert basis. Therefore, as induced trees of $C_4$ coincide with its induced path, Theorem~\ref{the:hbPath} is not a characterization of the graph for which this property holds. This result also raises the question of a characterization of the graphs for which induced paths' extended incidence vectors form Hilbert bases.
	
	\subsection{Induced paths of chordal graphs and integral systems}
	
	In this section, we discuss the consequences of Theorem~\ref{the:hbPath} and use it to give a description of the induced path polytope of chordal graphs.
	Similarly to what we do for the induced tree polytope, we exhibit a description of the induced path polytope of chordal graphs. We denote by $\mathcal{P}(G)$ the polyhedron defined as the intersection of $\mathcal{C}_P(G)$ and the hyperplane:
	
	\begin{equation}
		\{ (x,y) ~|~ x(V) - y(E) = 1\}. \label{iq:VboundEbb}
	\end{equation} 
	
	\begin{theorem}\label{the:inducedpathchordalinteger}
		A graph $G=(V,E)$ is chordal if and only if $\mathcal{P}(G)$ is binary.
	\end{theorem}
	\begin{proof}
		$(\Rightarrow)$
		The extreme points of $\mathcal{P}(G)$ are scalings of $\mathcal{C}_P(G)$'s generators which belong to~\eqref{iq:VboundEbb}. Theorem~\ref{the:hbPath} implies that the set of induced paths' extended incidence vectors generate $\mathcal{C}_P(G)$  moreover, they belong to \eqref{iq:VboundEbb}. Finally, by Observation~\ref{obs:contained} and Corollary~\ref{cor:inducedtreechordal}, $\mathcal{P}(G)$ is also bounded.
		
		$(\Leftarrow)$ The proof is by contradiction. The same non binary point exhibited in the proof of Theorem~\ref{the:inducedPathChordal} also belongs to $\mathcal{P}(G)$.
	\end{proof}

	Note that inequalities~\eqref{iq:trivialx} are mandatory only when the graph has isolated vertices. If the input graph has no such vertex, then inequalities~\eqref{iq:neighCliquebb} and~\eqref{iq:trivialYebb} ensure that $x$ is positive. The proof of Theorem~\ref{the:inducedpathchordalinteger} shows that the extreme points of $\mathcal{P}(G)$ are induced paths' extended incidence vectors, which implies the following.
	
	\begin{corollary}
		Given a graph $G$, the following are equivalent:
		\begin{itemize}
			\item $G$ is chordal
			\item $\mathcal{P}(G)$ is binary
			\item $\mathcal{P}(G)$ is bounded
		\end{itemize}   
	\end{corollary}
	
	It turns out that not all inequalities~\eqref{iq:neighCliquebb} define facets; we characterize which ones do.
	
	\begin{theorem}\label{the:facetorbit}
		Given a graph $G = (V,E)$, a vertex $w$, and a clique $K \subseteq N(w)$, the associated inequality~\eqref{iq:neighCliquebb} defines a facet if and only if no connected component of $\overline{G}[C_G(K\cup\{w\})]$ is bipartite.
	\end{theorem}
	\begin{proof}
		$(\Rightarrow)$ Let us consider the connected components $\overline{G}_1,\dots,\overline{G}_k$  of $\overline{G}[C_G(K \cup \{w\})]$ and suppose by contradiction that for $i \le \ell$, $\overline{G}_i$ is bipartite composed of the stable sets induced by $S_i^1, S_i^2$ for some $l \in \{1,\dots k\}$. By construction, $K_j = \bigcup_{i = 1}^\ell S_i^j \cup K \cup \{w\}, \  j = 1,2$ is a clique of $G$ such that $C_G(K_j ) = \bigcup_{i = \ell +1 }^k V(\overline{G}_i)$. 
		By summing the valid inequalities $2x_{w} \ge 2 y (\delta(w,K_j)) + y(\delta(w,C(K_j)))$ for $j = 1,2$ we obtain the inequality associated with $K$ and $w$.  This shows that it is dominated and hence not facet-defining.
		
		$(\Leftarrow)$
		We exhibit $|V| + |E|$ linearly independent points satisfying~\eqref{iq:neighCliquebb} with equality for a given $w$ and $K$. This is sufficient as by equality~\eqref{iq:VboundEbb} the dimension of $\mathcal{P}(G)$ is at most $|V| + |E|-1$. We consider $K_w = K \cup \{w\}$. The points are as follows:
		\begin{enumerate}
			\item $\xi^{\{v\}} \ \forall v \in V \setminus \{ w \}$  \label{point:chiv}
			\item $\xi^{\{u,v\}} \ \forall uv \in E \setminus \delta (w)$
			\item $\xi^{\{u,w\}} \ \forall uw \in \delta(w,K)$
			\item $\xi^{\{u,w,v\}} \ \forall uw \in \delta(w) \setminus \Biggl(\delta \Bigl(K \Bigr) \cup \delta \Bigl(C(K_w) \Bigr) \Biggr) $ for some vertex $v \in K$ non adjacent to $u$ \label{point:chilast}
		\end{enumerate}
		
		Points~\ref{point:chiv} to \ref{point:chilast} satisfy inequality~\eqref{iq:neighCliquebb} for $w$ and $K$ with equality. Moreover, they can be ordered so that the corresponding matrix is triangular. For the remaining $\left|\delta \Bigl(w,C(K_w)\Bigr)\right|$ elements, two edges in $\delta \Bigl(C(K_w) \Bigr)$ are needed to achieve equality; therefore, it is not possible to keep the triangular structure.
	
		The remaining $\left|\delta \Bigl(w,C(K_w)\Bigr)\right|$ points are built the following way: for each connected components $\overline{G}_c$ of $\overline{G}[C(K_w)]$, let us consider $T _c\subseteq E(\overline{G}_c)$ inducing a subgraph containing an odd cycle and hitting every vertex of $G_c$. $T_c$ exists and may be chosen such that $|T_c| = |V(\overline{G}_c)|$: start with an odd cycle and then add the missing edges for this set to span the vertices. Now, match each edge of $\delta(w,C(K_w))$ to a distinct edge of $T = \bigcup T_c$ sharing an extremity, note that $|T| = |C(K_w)|=|\delta(w,C(K_w))|$. Let $uw$ be in  $\delta(w,C(K_w))$ and  $uv$ be its corresponding edge of $T$. In $G$, $u$ and $v$ are non-adjacent, hence $\xi^{u,v,w}$ is the incidence vector of an induced path for which inequality~\eqref{iq:neighCliquebb} for $w$ and $K$ is tight. 
		We refer to Figure~\ref{fig:matrixOfPoints} for an illustration of the matrix whose rows are the points previously exhibited and in the same order, therefore, the lines containing the submatrix $H$ correspond to the last set of $\left|\delta \Bigl(w,C(K_w)\Bigr)\right|$  points.  The symbol "?" stands for the submatrices associated with incidence vectors of 3-vertex paths and whose value does not need to be explicit for the validity of the proof. 
		
		 $H$ is identical to the incidence matrix of the non-bipartite graph induced by $T$, by Theorem 2.1 of~\cite{GROSSMAN1995213} it has full row rank $|C(K_w)|$.
		
		Full row rank of the whole matrix is obtained by the facts that the matrix is block triangular with a non-zero diagonal over the points~\eqref{point:chiv}--\eqref{point:chilast}, and that $H$ is the maximal submatrix with no 0 row having  elements of $\delta \Bigl(w,C(K_w)\Bigr)$ as column set. This implies that the whole matrix has rank $|V| + |E|$, therefore the inequality is facet defining.
	\end{proof}
	
	\begin{figure}
		\centering
		\begin{tikzpicture}
			\node[rotate = 45] at (-4.75,2.3) {\small$V\setminus w$};
			\node[rotate = 45] at (-2.75,2.5) {\small$E\setminus \delta(w)$};
			\node[rotate = 45] at (-0.5,2.6) {\small $\delta(w,K)$};
			\node[rotate = 45] at (2.75,3.2) {\small $\delta(w) \setminus \delta(K) \setminus \delta (C(K\cup \{w\}) ) $};
			\node[rotate=45] at (5,2.6) {\small $\delta(w,C(K \cup \{w\}))$};
			
			\node at (0,0) {
				$\left(\begin{array}{clllllllcl}
					\hline
					I_{V\setminus w}&|&0 \quad \dots&& && \\
					\hline
					\multirow{2}{*}{$\mathcal{E}(G')$}&|&I_{E\setminus \delta(w)}&|&0 \quad \dots &&\\
					\cmidrule{3-9}
					& |&0\quad \dots &|&I_{\delta(w,K)}&| &0 \quad  \dots&\\
					\hline
					&\multicolumn{4}{c}{\multirow{1.5}{*}{?} }&|&I_{\delta(w) \setminus \delta(K) \setminus \delta (C(K))  }&|&0 \quad \dots \\
					\hline
					?& |&0 \quad \dots & &&&& |&H \\
					\hline
				\end{array} \right)$ 
			}
			;
		\end{tikzpicture}
		\caption{Matrix of points built in the proof of Theorem~\ref{the:facetorbit} where $G'$ is the graph with vertex set $V$ and edge set $(E\setminus \delta(w) )\cup (\delta(w) \cap \delta(K)) $ }
		\label{fig:matrixOfPoints}
	\end{figure}
	
	The characterization of Theorem~\ref{the:facetorbit} gives no information about the number of such facet-defining inequalities, we show that they can be separated in polynomial time by a reduction to a min cut problem in an oriented graph.

	\begin{theorem}\label{the:separationofIQ}
		For chordal graphs, Inequalities~\eqref{iq:neighCliquebb} can be separated in polynomial time.
	\end{theorem}
	\begin{proof}
		Let $(x^*,y^*)$ be a point of $\mathcal{P}(G)$.
		We first iterate over each maximal clique $K$ and vertex $v$ of $K$, as the input graph is chordal this is polynomial. Observe that the edges involved in a constraint~\eqref{iq:neighCliquebb} are all incident to $v$, for this reason, consider the weight $m_{u} = y_{uv}^*$ for each neighbor $u$ of $v$. If a most violated Inequality~\eqref{iq:neighCliquebb} is associated with a subset of $K$, $K\setminus K'$ is a set of complete vertices to $K'$, and the other are among  $N(K) \cap N(v) = W$. Therefore, for fixed $w$ and $K$, the separation problem is to find the subset of $K'$ maximizing $2m(K') + m(K\setminus K') + m(C(K')\cap W )$. Let us first cast this problem into a minimization one by observing that the sum of the weight $m$ is constant over $K$ and over $W$ that is to minimize: $m( K\setminus K') + m(W \setminus C(K'))$. Let us build a directed graph $D$ with vertices $\{s,t\} \cup K \cup W$. Let all vertices $w$ of $W$ have a directed arc towards $t$ with weight $m_w$ and let $s$ have an arc towards all vertices $k$ of $K$ with weight $m_k$. Now, add an oriented arc with infinite weight from each vertex of $k \in K$ towards each vertex $w\in W$ such that $kw$ is not an edge of $G$. Let $\delta(Q)$ be a minimum $s-t$ cut of $D$ where $s\in Q$ and $K' =  K \cap Q$, $W' = W \cap Q$. 
		
		The vertices of $K'$ must be though as the subset of $K$ associated with inequality~\eqref{iq:neighCliquebb} and the vertices of $W'$ must be though as the set of vertices not complete to $K'$. 
		Notice that their exist $s$-$t$ cuts containing no infinite weight arcs, either $\delta(s)$ or $\delta(s \cup K \cup W)$, those respectively correspond to the inequalities $y(\delta(w)) \le 2x_v$ and $y(\delta(K,w)) \le x_v$.

		Infinite weight arcs ensure that no vertex not complete to $K'$ belongs to the shore of $t$, this ensures that the vertices $w$ not complete to $K'$ contributes to the objective function by $m_w$. 
		
		We verify that the objective function matches.
		Observe that for fixed $K'$, for every vertex $w$ complete to $K'$, there is no arc going from $K'$ towards $w$, and therefore $w$ must belong to the shore of $t$ to minimize the objective function.		
		Finally, the vertices $k$ of $K\setminus K'$ contributes by $m_v$ to the objective function.
		Notice that the infinite weight directed arc going from $K\setminus K'$ towards $W'$ do not contribute to the cost of $\delta(Q)$ as they are not directed in the correct sense. This shows that the cost of $\delta(Q)$ is exactly $ m(K\setminus K') + m(W')$.
		
		The inequality associated with $K'$ and $v$ is violated if $ m(K\setminus K') + m(W') > 2x_v^*$. 		
	\end{proof}
	
Theorem~\eqref{the:separationofIQ} and equivalence between separation and optimization~\cite{separationOptimization} via the ellipsoid method implies the following.

	\begin{corollary}
		The maximum vertex and weighted induced path problem is solvable in polynomial time in chordal graphs.
	\end{corollary}

	\section*{Conclusion}
	
	In this article, we show that the induced tree and path extended incidence vectors of chordal graphs respectively, form two Hilbert bases. We use this fact to give two linear systems of inequalities for which the binarity of integer points and boundedness characterizes chordality. It turns out that there exist non-chordal graphs for which induced tree and path extended incidence vectors form a Hilbert basis, which raises the question of characterizing such graphs. 
	These results imply that both corresponding linear optimization problems are polynomially solvable.
	In general, showing that some subsets of vectors of a ground set $\mathcal{V}$ forms Hilbert bases is used to show that the linear system whose rows are elements of $\mathcal{V}$ is TDI. In this work, we do not consider such systems and observe that extended incidence vectors of induced trees and path form Hilbert bases in chordal graphs. It is still a question whether these new Hilbert basis can be of any use to find new TDI systems. Moreover, our work does not characterize the graphs for which extended vectors of induced paths and induced trees form Hilbert bases, we have no conjecture on an induced subgraph characterization of those. However, as a first step to address this question we propose the following conjecture.
	
	\begin{conjecture}
The induced path's (resp induced tree) extended incidence vectors of of odd hole's complements do not form Hilbert bases and this property does not hold for their induced subgraphs.
\end{conjecture}  
	
	\section*{Statements and Declarations}
	
	The author thanks the ANR project ANR-24-CE23-1621 EVARISTE for financial support and declares that he has no conflicts of interest. 

	\bibliographystyle{cas-model2-names}
	
	\bibliography{bibtex}
\end{document}